\documentclass[12pt]{article}
\usepackage[latin1]{inputenc}
\usepackage{cmap}
\usepackage{lmodern}

\usepackage{amssymb, amsmath, amsthm}
\usepackage[a4paper,top=25mm,bottom=25mm,left=25mm,right=25mm]{geometry}
\usepackage{etex}

\usepackage{authblk} 
\usepackage{pifont}
\usepackage{graphicx}
\usepackage[usenames,dvipsnames,svgnames,table]{xcolor}
\usepackage[figuresright]{rotating}
\usepackage{xtab} 
\usepackage{longtable} 
\usepackage{multirow}
\usepackage{footnote}
\usepackage[stable]{footmisc}
\usepackage{chngpage} 
\usepackage{pdflscape} 

\usepackage{pgfplots}
\usepackage{setspace}

\makesavenoteenv{tabular}
\usepackage{tabularx}
\usepackage{booktabs}
\usepackage{threeparttable}
\usepackage[referable]{threeparttablex} 
\newcolumntype{R}{>{\raggedleft\arraybackslash}X}
\newcolumntype{L}{>{\raggedright\arraybackslash}X}
\newcolumntype{C}{>{\centering\arraybackslash}X}
\newcolumntype{A}{>{\columncolor{gray!25}}C}
\newcolumntype{a}{>{\columncolor{gray!25}}c}

\usepackage{dcolumn} 
\newcolumntype{.}{D{.}{.}{-1}}

\usepackage{tikz}
\usetikzlibrary{arrows}
\usepackage[semicolon]{natbib}
\usepackage[hyphens]{url}
\usepackage[bookmarks=true]{hyperref} 
\hypersetup{
  colorlinks   = true,    
  urlcolor     = blue,    
  linkcolor    = blue,    
  citecolor    = red      
}
\usepackage{microtype}
\usepackage[justification=centerfirst]{caption}

\usepackage[labelformat=simple]{subcaption}

\DeclareCaptionLabelFormat{parenthesis}{(#2)}
\makeatletter
\renewcommand\p@subfigure{\arabic{figure}.}
\makeatother

\DeclareCaptionLabelFormat{parenthesis}{(#2)}
\makeatletter
\renewcommand\p@subtable{A.\arabic{table}.}
\makeatother

\usepackage{enumitem}

\setlist[itemize]{leftmargin=3\parindent}
\setlist[enumerate]{leftmargin=2\parindent}

\theoremstyle{plain}

\newtheorem{lemma}{Lemma}
\newtheorem{proposition}{Proposition}
\newtheorem{theorem}{Theorem}

\theoremstyle{definition}

\newtheorem{definition}{Definition}
\newtheorem{example}{Example}

\theoremstyle{remark}
\newtheorem{notation}{Notation}
\newtheorem{remark}{Remark}

\newcommand{\down}{\textcolor{BrickRed}{\rotatebox[origin=c]{270}{\ding{212}}}}
\newcommand{\up}{\textcolor{PineGreen}{\rotatebox[origin=c]{90}{\ding{212}}}}

\def\keywords{\vspace{.5em} 
{\noindent \textit{Keywords}:\,}}

\def\JEL{\vspace{.5em} 
{\noindent \textbf{\emph{JEL} classification number}:\,}}

\def\AMS{\vspace{.5em} 
{\noindent \textbf{\emph{AMS} classification number}:\,}}

\author{L\'aszl\'o Csat\'o\thanks{~e-mail: laszlo.csato@uni-corvinus.hu} }
\affil{Institute for Computer Science and Control, Hungarian Academy of Sciences (MTA SZTAKI) \\
Laboratory on Engineering and Management Intelligence, Research Group of Operations Research and Decision Systems}
\affil{Corvinus University of Budapest (BCE) \\
Department of Operations Research and Actuarial Sciences}
\affil{Budapest, Hungary}
\title{On the ranking of a Swiss system chess team tournament}
\date{\today}

\begin{document}

\maketitle

\begin{abstract}
The paper suggests a family of paired comparison-based scoring procedures for ranking the participants of a Swiss system chess team tournament. We present the challenges of ranking in Swiss system, the features of individual and team competitions as well as the failures of the official rankings based on lexicographical order. The tournament is represented as a ranking problem such that the linearly-solvable row sum (score), generalized row sum, and least squares methods have favourable axiomatic properties.

Two chess team European championships are analysed as case studies. Final rankings are compared by their distances and visualized with multidimensional scaling (MDS). Differences to the official ranking are revealed by the decomposition of the least squares method. Rankings are evaluated by prediction power, retrodictive performance, and stability. The paper argues for the use of least squares method with a results matrix favouring match points on the basis of its relative insensitivity to the choice between match and board points, retrodictive accuracy, and robustness.

\JEL{D71}

\AMS{15A06, 91B14}

\end{abstract}


\section{Introduction} \label{Sec1}

Chess tournaments are often organized in the Swiss system when there are too many participants to play a round-robin tournament. They go for a predetermined number of rounds, in each round two players compete head-to-head. None of them are eliminated during the tournament, but there are pairs of players without a match between them.
Let us denote the number of rounds by $c$ and the number of participants by $n$.

Two emerging issues in Swiss system tournaments are how to pair the players and how to rank the participants on the basis of their respective results. The pairing algorithm aims to pair players with a similar performance, measured by the number of their wins and draws (see \citet{FIDE2015} for details).

The ranking involves two main challenges. The first one is the possible appearance of \emph{circular triads} when player $i$ has won against player $j$, player $j$ has won against player $k$, but player $k$ has won against player $i$.
The second issue arises as the consequence of incomplete comparisons since $c < n-1$. For example, if player $i$ has played only against player $j$, then its rank probably should depend on the results of player $j$.

The final ranking of the players is usually determined by the aggregated number of points scored: the winner of a match gets one point, the loser gets zero points, and a draw means half-half points for both players.
However, it usually does not result in a linear order (a complete, transitive and antisymmetric binary relation) of the participants.\footnote{~In $c$ rounds the number of match points can be an integer or a half-integer between $0$ and $c$, so there always will be players with equal score if $n > 2c+1$.}
Ties are eliminated by the sequential application of various tie-breaking rules \citep{FIDE2015}.

These ranking(s), based on lexicographical orders, will be called \emph{official ranking(s)}. They can differ in the tie-breaking rules.

Official rankings are not able to solve the problem caused by different schedules as players with weaker opponents can score the same number of points more easily \citep{Brozos-Vazquezetal2010, Csato2012a, Csato2013a, Forlano2011, JeremicRadojicic2010, Redmond2003}.
It turns out that players with an improving performance during the tournament are preferred contrary to players with a declining one. Consider two players $i$ and $j$ with an equal number of points after playing some rounds. Player $i$ is said to be on the \emph{inner circle} if it scored more points in some of the first rounds relative to player $j$ who is said to be on the \emph{outer circle}. Since they have played against opponents with a similar number of points in each round, it is probable that player $j$ has met with weaker opponents. Tie-breaking rules may take the performance of opponents into account but a similar problem arises if player $j$ has marginally more points than player $i$ as a lexicographical order is not continuous.

This is known to be an inherent defect of these systems. In fact, players sometimes may deliberately seek for a draw or defeat in the first \citep{Forlano2011}. It is tolerated because the case concerns very rarely the winner of the tournament, at least with an adequate number of rounds.

Nevertheless, some works have aimed to improve on ranking in Swiss system tournaments.
\citet{Redmond2003} presents a generalization of win-loss ratings by accounting for the strength of schedule. \citet{Brozos-Vazquezetal2010} argue for the use of recursive methods (as a tie-breaking rule) in Swiss system chess tournaments. \citet{Forlano2011} shows a way to correct the points for wins and draws in order to derive a more legitimate ranking.

The current paper will attempt to solve the problem through the use of paired comparison-based ranking procedures. In a sense, it is a return to the origins of this line of research as early works were often inspired by chess tournaments \citep{Landau1895, Landau1914, Zermelo1929}.

We use a parametric family of scoring procedures based on linear algebra, the generalized row sum \citep{Chebotarev1989_eng, Chebotarev1994} as well as the least squares method, which was extensively used for sport rankings \citep{LeeflangvanPraag1971, Stefani1980}.
Despite the issue has been touched by \citet{Csato2013a}, here a deeper methodological foundation will be given for the problem and the evaluation of rankings will be revisited.

The analysis is based on some recent results. \citet{Gonzalez-DiazHendrickxLohmann2013} have discussed the axiomatic properties of generalized row sum and least squares, \citet{Csato2015a} has given an interpretation for the least squares method, while \citet{Can2014} has contributed to the choice of distance functions between rankings.

In order to avoid the prominent role of colour allocation in individual championships, the discussion focuses on \emph{team tournaments}, where a match between two teams is played on $2t$ boards such that $t$ players of a team play with white and the other $t$ players of the team play with black. The winner of a game on a board gets $1$ point, the loser gets $0$ points, and the draw yields $0.5$ points for both teams, thus $2t$ \emph{board points} are allocated in a given match. The winner team achieving more (at least $t + 0.5$) board points scores $2$ \emph{match points}, the loser $0$, while a draw results in $1$ match point for both teams. Therefore one can choose between the board points and the match points as the basis of the official ranking. Recently the use of match points is preferred in chess olympiads and team European championships.

The paper is structured as follows. Section~\ref{Sec2} shortly outlines the ranking problem, ranking methods and their relevant properties. It also aims to incorporate Swiss system chess team tournaments into this framework.
The proposed model is applied in Section~\ref{Sec3} to rank the participants in the 2011 and 2013 European Team Chess Championship open tournaments. Twelve rankings, distinguished by the influence of opponents' performance and match versus table points, are introduced and compared on the basis of their distances. They are visualized with multidimensional scaling (MDS), while differences to the official ranking are revealed by the decomposition of the least squares method.

On the basis of these results, we argue for the use of least squares method with a generalized result matrix favouring match points. It is supported by several arguments, variability of the ranking with respect to the role of match and board points as well as retrodictive performance (the ability to match the outcomes of matches already played) and robustness (stability of the ranking between two subsequent rounds). 

Finally, Section~\ref{Sec5} summarizes the findings and review possible extensions of the model.
A reader familiar with ranking problems \citep{Gonzalez-DiazHendrickxLohmann2013, Csato2015a} may skip Subsections~\ref{Sec21} and \ref{Sec22}.

\section{Modelling of the tournament as a ranking problem} \label{Sec2}

In the following some fundamental concepts of the paired comparison-based ranking methodology are presented. A detailed discussion can be found in \citet{Csato2015a}.

\subsection{The ranking problem} \label{Sec21}

Let $N = \{ 1,2, \dots ,n \}$, $n \in \mathbb{N}$ be a set of objects. The \emph{matches matrix} $M = (m_{ij}) \in \mathbb{N}^{n \times n}$ contains the number of comparisons between the objects, and is symmetric ($M^\top = M$).\footnote{~In most practical applications (including ours) the condition $m_{ij} \in \mathbb{N}$ means no restriction. Modification of the domain to $\mathbb{R}_+$ has no impact on the results but the discussion becomes more complicated. This generalization has some significance for example in the case of forecasting sport results when the latest comparisons may give more information about the current form of a player.}
Diagonal elements $m_{ii}$ are supposed to be $0$ for all $i = 1,2, \dots ,n$. Let $d_i = \sum_{j=1}^n m_{ij}$ be the \emph{total number of comparisons} of object $i$ and $\mathfrak{d} = \max \{  d_i: i \in N \}$ be the \emph{maximal number of comparisons} with the other objects.
Let $m = \max \{ m_{ij}: i,j \in N \}$.

The \emph{results matrix} $R = (r_{ij}) \in \mathbb{R}^{n \times n}$ contains the outcome of comparisons between the objects, and is skew-symmetric ($R^\top = -R$). All elements are limited by $r_{ij} \in \left[ -m_{ij},\, m_{ij} \right]$. $(r_{ij} + m_{ij})/(2m_{ij}) \in \left[ 0,1 \right]$ may be regarded as the likelihood that object $i$ defeats $j$.
Then $r_{ij} = m_{ij}$ means that $i$ is perfectly better than $j$, and $r_{ij} = 0$ corresponds to an undefined relation (if $m_{ij} = 0$) or to the lack of preference (if $m_{ij} > 0$) between the two objects.
A \emph{ranking problem} is given by the triplet $(N,R,M)$. Let $\mathcal{R}^n$ be the class of ranking problems with $|N| = n$.

A ranking problem is called \emph{round-robin} if $m_{ij} = 1$ for all $i \neq j$, that is, every object has been compared exactly once with all of the others.
A ranking problem is called \emph{balanced} if $d_{i} = d_{j}$ for all $i,j = 1,2, \dots ,n$, that is, every object has the same number of comparisons.

\subsection{Rankings derived from rating functions} \label{Sec22}

Matches matrix $M$ can be represented by an undirected multigraph $G := (V,E)$ where vertex set $V$ corresponds to the object set $N$, and the number of edges between objects $i$ and $j$ is equal to $m_{ij}$. The number of edges adjacent to $i$ is the \emph{degree} $d_i$ of node $i$. A \emph{path} from object $k_1$ to object $k_t$ is a sequence of objects $k_1, k_2, \dots , k_t$ such that $m_{k_\ell k_{\ell+1}} > 0$ for all $\ell = 1,2, \dots ,t-1$. Two vertices are \emph{connected} if $G$ contains a path between them. A graph is said to be connected if every pair of vertices is connected.

Graph $G$ is called the \emph{comparison multigraph} associated with the ranking problem $(N,R,M)$, and is independent of the results of paired comparisons. The \emph{Laplacian matrix} $L = (\ell_{ij}) \in \mathbb{R}^{n \times n}$ of graph $G$ is given by $\ell_{ij} = -m_{ij}$ for all $i \neq j$ and $\ell_{ii} = d_i$ for all $i = 1,2, \dots ,n$.

Vectors are denoted by bold fonts, and assumed to be column vectors.
Let $\mathbf{e} \in \mathbb{R}^n$ be given by $e_i = 1$ for all $i = 1,2, \dots ,n$ and $I \in \mathbb{R}^{n \times n}$ be the matrix with $I_{ij} = 1$ for all $i,j = 1,2, \dots ,n$.

A \emph{rating (scoring) method} $f$ is an $\mathcal{R}^n \to \mathbb{R}^n$ function, $f_i = f_i(N,R,M)$ is the rating of object $i$. It defines a \emph{ranking method} by $i$ weakly above $j$ in the ranking problem $(N,R,M)$ if and only if $f_i(N,R,M) \geq f_j(N,R,M)$.
Throughout the paper, the notions of rating and ranking methods will be used analogously since all ranking procedures discussed are based on rating vectors. Rating methods $f^1$ and $f^2$ are called \emph{equivalent} if they result in the same ranking for any ranking problem $(N,R,M)$.

Let us introduce some rating methods.

\begin{definition} \label{Def1}
\emph{Row sum} method:
$\mathbf{s}(N,R,M): \mathcal{R}^n \to \mathbb{R}^n$ such that $\mathbf{s} = R \mathbf{e}$.
\end{definition}

Row sum will also be referred to as \emph{scores}, $\mathbf{s}$ is sometimes called the scores vector. It does not take the comparison multigraph into account.

\begin{definition} \label{Def2}
\emph{Generalized row sum} method:
$\mathbf{x}(\varepsilon)(N,R,M): \mathcal{R}^n \to \mathbb{R}^n$ such that
\[
(I+ \varepsilon L) \mathbf{x}(\varepsilon) = (1 + \varepsilon m n)\mathbf{s},
\]
where $\varepsilon > 0$ is a parameter.
\end{definition}

This parametric rating procedure was constructed axiomatically by \citet{Chebotarev1989_eng} and thoroughly analysed in \citet{Chebotarev1994}.
Generalized row sum adjusts the standard row sum by accounting for the performance of objects compared with it, and so on. $\varepsilon$ indicates the importance attributed to this correction.
It follows from the definition that $\lim_{\varepsilon \to 0} \mathbf{x}(\varepsilon) = \mathbf{s}$ for all ranking problems $(N,R,M)$.

Both the row sum and the generalized row sum ratings are well-defined and can be obtained from a system of linear equations for all ranking problems since $(I+ \varepsilon L)$ is positive definite for any $\varepsilon \geq 0$.

In our model the outcome of paired comparisons is restricted by $-m \leq r_{ij} \leq m$ for all $i,j \in N$. Then \citet[Proposition~5.1]{Chebotarev1994} argues that the \emph{reasonable upper bound} of $\varepsilon$ is $1 / \left[ m(n-2) \right]$.

$h_{ij} = r_{ij} / \min \{ m_{ij}; 1 \} \in \left[ -1, 1 \right]$\footnote{~$\min \{ m_{ij}; 1 \}$ is written in order to avoid division by zero.} may be identified as the normalized difference between the latent ratings $q_i$ and $q_j$ of objects $i$ and $j$. Then it makes sense to choose $\mathbf{q}$ in order to minimize the error according to an appropriate objective function.

\begin{definition} \label{Def4}
\emph{Least squares} method: $\mathbf{q}(N,R,M): \mathcal{R}^n \to \mathbb{R}^n$ such that it is the solution to the problem
\[
\min_{\mathbf{q} \in \mathbb{R}^n} m_{ij} \left[ h_{ij} - \left( q_i - q_j \right) \right]^2
\]
satisfying $\mathbf{e}^\top \mathbf{q} = 0$.
\end{definition}

The normalization $\mathbf{e}^\top \mathbf{q} = 0$ is necessary because the value of the objective function is the same for $\mathbf{q}$ and $\mathbf{q} + \beta \mathbf{e}$, $\beta \in \mathbb{R}$.

The least squares ranking method is well-known in a lot of fields, a review about its origin is given by \citet{Gonzalez-DiazHendrickxLohmann2013} and \citet{Csato2015a}.
It has strong connections to generalized row sum.

\begin{proposition} \label{Prop2}
The least squares rating can be obtained as a solution of the linear system of equations $L \mathbf{q} = \mathbf{s}$ and $\mathbf{e}^\top \mathbf{q} = 0$ for all ranking problems $(N,R,M)$.
\end{proposition}

\begin{proof}
See \citet[p.~57]{Csato2015a}.
\end{proof}

\begin{lemma} \label{Lemma1}
For all ranking problems $(N,R,M)$, the least squares method is equivalent to the limit of generalized row sum if $\varepsilon \to \infty$ since $\lim_{\varepsilon \to \infty} \mathbf{x}(\varepsilon) = m n \mathbf{q}$.
\end{lemma}

\begin{proof}
See \citet[p.~326]{ChebotarevShamis1998a} and \citet{Csato2016a}.
\end{proof}

\begin{proposition} \label{Prop3}
The least squares rating $\mathbf{q}(N,R,M)$ is unique if and only if comparison multigraph $G$ of the ranking problem is $(N,R,M)$ connected.
\end{proposition}

\begin{proof}
See \citet[p.~59]{Csato2015a}.
\end{proof}

Note that in the case of an unconnected comparison multigraph there are independent ranking problems.

A graph-theoretic interpretation of the generalized row sum method is given by \citet{Shamis1994}.
\citet{Csato2015a} provides the following iterative decomposition of least squares. 

\begin{proposition} \label{Prop4}
Let the comparison multigraph of a ranking problem $(N,R,M)$ be connected and not regular bipartite. Then the unique solution of the least squares problem is $\mathbf{q} = \lim_{k \to \infty} \mathbf{q}^{(k)}$ where
\[
\mathbf{q}^{(0)} = (1/ \mathfrak{d}) \mathbf{s},
\]
\[
\mathbf{q}^{(k)} = \mathbf{q}^{(k-1)} + \frac{1}{\mathfrak{d}} \left[ \frac{1}{\mathfrak{d}} \left( \mathfrak{d}I - L \right) \right]^k \mathbf{s} \quad (k = 1,2, \dots).
\]
\end{proposition}

\subsection{Two properties of scoring procedures} \label{Sec23}

In order to argue for the use of these methods we discuss some axioms.

\begin{definition} \label{Def5}
\emph{Admissible transformation of the results} \citep{Csato2014b}:
Let $(N,R,M) \in \mathcal{R}^n$ be a ranking problem. An \emph{admissible transformation of the results} provides a ranking problem $(N,kR,M) \in \mathcal{R}^n$ such that $k > 0$, $k \in  \mathbb{R}$ and $k r_{ij} \in \left[ -m_{ij},m_{ij} \right]$ for all $i \in N$.
\end{definition}

Multiplier $k$ cannot be too large since $-m_{ij} \leq k r_{ij} \leq m_{ij}$ should be satisfied for all $i,j \in N$ according to the definition of the results matrix. $k \leq 1$ is always allowed.

\begin{definition} \label{Def6}
\emph{Scale invariance} ($SI$) \citep{Csato2014b}:
Let $(N,R,M), (N,kR,M) \in \mathcal{R}^n$ be two ranking problems such that $(N,kR,M)$ is obtained from $(N,R,M)$ through an admissible transformation of the results.
Scoring procedure $f: \mathcal{R}^n \to \mathbb{R}^n$ is \emph{scale invariant} if $f_i(N,R,M) \geq f_j(N,R,M) \Leftrightarrow f_i(N,kR,M) \geq f_j(N,kR,M)$ for all $i,j \in N$.
\end{definition}

Scale invariance implies that the ranking is invariant to a proportional modification of wins ($r_{ij} > 0$) and losses ($r_{ij} < 0$). It seems to be important for applications. If the outcomes of paired comparisons cannot be measured on a continuous scale, it is not trivial how to transform them into $r_{ij}$ values. $SI$ provides that it is not a problem in certain cases. For example, if only three outcomes are possible, the coding ($r_{ij} = \kappa$ for wins; $r_{ij} = 0$ for draws; $r_{ij} = -\kappa$ for losses) makes the ranking independent from $0 < \kappa \leq 1$. It may also be advantageous when relative intensities are known such as a regular win is two times better than an overtime triumph.

\begin{lemma} \label{Lemma2}
The row sum, generalized row sum and least squares methods satisfy $SI$.
\end{lemma}

\begin{proof}
See \citet[Lemma~4.3]{Csato2014b}.
\end{proof}

One disadvantage of the row sum procedure is its independence of irrelevant matches \citep{Gonzalez-DiazHendrickxLohmann2013, Csato2015d}.
However, it causes no problem in the round-robin case, so it makes sense to preserve the attributes of row sum on this set.  

\begin{definition} \label{Def7}
\emph{Score consistency} ($SCC$) \citep{Gonzalez-DiazHendrickxLohmann2013}:
Scoring procedure $f: \mathcal{R}^n \to \mathbb{R}^n$ is \emph{score consistent} if $f_i(N,R,M) \geq f_j(N,R,M) \Leftrightarrow s_i(N,R,M) \geq s_j(N,R,M)$ for all $i,j \in N$ and round-robin ranking problem $(N,R,M) \in \mathcal{R}^n$.
\end{definition}

A score consistent method is equivalent to the row sum method in the case of round-robin ranking problems. A similar requirement is mentioned by \citet{Zermelo1929} and \citet[Property~3]{David1987}.

\begin{remark} \label{Rem1}
Regarding the generalized row sum method, \citet[Property~3]{Chebotarev1994} introduces a more general axiom called \emph{agreement}: if $(N,R,M) \in \mathcal{R}^n$ is a round-robin ranking problem, then $\mathbf{x}(\varepsilon)(N,R,M) = \mathbf{s}(N,R,M)$.
\end{remark}

\begin{lemma} \label{Lemma3}
Row sum, generalized row sum and least squares methods satisfy $SCC$.
\end{lemma}

\begin{proof}
For generalized row sum, see Remark~\ref{Rem1}.
In the case of least squares the proof is given by \citet[Proposition~5.3]{Gonzalez-DiazHendrickxLohmann2013}.
\end{proof}

Further properties of the scoring procedures are discussed by \citet{Gonzalez-DiazHendrickxLohmann2013} and \citet{Csato2014b}.

\subsection{Interpretation of Swiss system chess team tournaments as a ranking problem} \label{Sec24}

In order to use the scoring procedures presented above, the chess tournament should be formulated as a ranking problem:
\begin{itemize}[label=$\bullet$]
\item
Set of objects $N$ consists of the teams of the competition;
\item
Matches matrix $M$ is given by $m_{ij} = 1$ if teams $i$ and $j$ have played against each other and $m_{ij} = 0$ otherwise.
\end{itemize}

For the sake of simplicity it is assumed that $n$ is even, thus all teams play exactly $c$ matches (there are no byes\footnote{~A bye is a team which does not play a match in a given round.}).

The results matrix should be skew-symmetric. It excludes the incorporation of some individual competitions where a win results in three points and a draw gives one point since then a win and a loss is not equal to two draws. Furthermore, the model is not able to reflect that the first-mover with white have an inherent advantage in the game. In order to eliminate the role of colour allocation, only team tournaments are discussed.

First two extreme possibilities are suggested for the choice of results matrix.

\begin{notation} \label{Not1}
$MP_{ij}$ and $BP_{ij}$ are the numbers of match points and board points of team $i$ against team $j$, respectively. \\
$\mathbf{mp}$ and $\mathbf{gp}$ are the vectors of match points and board points, respectively.
\end{notation}

\begin{definition}
\emph{Match points ranking}:
The ranking derived from the vector $\mathbf{mp}$.
\end{definition}

\begin{definition}
\emph{Board points ranking}:
The ranking derived from the vector $\mathbf{bp}$.
\end{definition}

Note that match and board points rankings are usually not linear orders in a Swiss system tournament. Ties are broken according to the rules given by the official ranking.

\begin{definition} \label{Def8}
\emph{Match points based results matrix}:
Results matrix $R^{MP}$ is \emph{based on match points} if $r_{ij}^{MP} = MP_{ij}-1$ for all $i,j \in N$.
\end{definition}

\begin{definition} \label{Def9}
\emph{Board points based results matrix}:
Results matrix $R^{BP}$ is \emph{based on board points} if $r_{ij}^{BP} = BP_{ij}-t$ for all $i,j \in N$.
\end{definition}

Note that match and board points based results matrices are skew-symmetric (a match between two teams is played on $2t$ boards).
The two concepts can be integrated.

\begin{definition} \label{Def10}
\emph{Generalized results matrix}:
Results matrix $R^{P}(\lambda)$ is \emph{generalized} if $r_{ij}^{P}(\lambda) = (1-\lambda) \left( MP_{ij}-1 \right) + \lambda \left( BP_{ij}-t \right) / t$ for all $i,j \in N$ such that $\lambda \in \left[ 0,1 \right]$.
\end{definition}

\begin{lemma} \label{Lemma4}
$R^{P}(\lambda = 0) = R^{MP}$ and $R^{P}(\lambda = 1) = R^{BP}$.
\end{lemma}

The row sum method is closely related to the match and board points rankings.

\begin{lemma} \label{Lemma5}
Row sum method applied on the match points based results matrix is equivalent to the match points ranking: $s_i(R^{MP}) \geq s_i(R^{MP}) \iff mp_i \geq mp_j$.
\end{lemma}

\begin{proof}
$d_i = c$ for all $i \in N$, hence $\mathbf{s}(N,R^{MP},M) = \mathbf{mp} - c \mathbf{e}$.
\end{proof}

\begin{lemma} \label{Lemma6}
Row sum method applied on the board points based results matrix is equivalent to the board points ranking: $s_i(R^{BP}) \geq s_i(R^{BP}) \iff bp_i \geq bp_j$.
\end{lemma}

\begin{proof}
$d_i = c$ for all $i \in N$, hence $\mathbf{s}(N,R^{BP},M) = \mathbf{bp} - ct \mathbf{e}$.
\end{proof}

A crucial argument for the application of paired comparison-based ranking methodology is provided by the following result.

\begin{theorem} \label{Theo1}
Let $(N,R,M) \in \mathcal{R}^n$ be a round-robin ranking problem. Then:
\begin{itemize}[label=$\bullet$]
\item
Generalized row sum and least squares methods applied on the match points based results matrix are equivalent to the match points ranking: \\
$x_i(\varepsilon)(R^{MP}) \geq x_i(\varepsilon)(R^{MP}) \iff q_i(R^{MP}) \geq q_j(R^{MP}) \iff mp_i \geq mp_j$.
\item
Generalized row sum and least squares methods applied on the board points based results matrix are equivalent to the board points ranking: \\
$x_i(\varepsilon)(R^{BP}) \geq x_i(\varepsilon)(R^{BP}) \iff q_i(R^{BP}) \geq q_j(R^{BP}) \iff bp_i \geq bp_j$.
\end{itemize}
\end{theorem}

\begin{proof}
In the case of round-robin problems, generalized row sum and least squares are equivalent to the row sum method due to axiom $SCC$ (Lemma~\ref{Lemma3}), hence Lemmata~\ref{Lemma5} and \ref{Lemma6} provide the statement.
\end{proof}

Generalized row sum and least squares methods take the opponents of each team into account. Due to Theorem~\ref{Theo1}, they result in the official ranking without tie-breaking rules in the ideal round-robin case. When the official ranking is based on match points, the transformation $R^{MP}$ is recommended. Generalized results matrix with a small (i.e. close to $0$) parameter $\lambda$ gives a similar outcome but it reflects the number of board points, the magnitude of wins and losses. This effect becomes more significant as $\lambda$ increases. $R^{BP}$ extends the board points ranking to Swiss system competitions.

\begin{proposition} \label{Prop5}
Let $(N,R,M) \in \mathcal{R}^n$ be a ranking problem, and $k \in \left( 0,1 \right]$.
Generalized row sum and least squares methods give the same ranking if they are applied on $R^{MP}$ and $k R^{MP}$, on $R^{BP}$ and $k R^{BP}$ as well as on $R^{P}(\lambda)$ and $k R^{P}(\lambda)$.
\end{proposition}

\begin{proof}
It is the consequence of property $SI$ (Lemma~\ref{Lemma1}).
\end{proof}

Due to Proposition~\ref{Prop5}, there exists only one ranking on the basis of match points if wins are more valuable than losses and draws correspond to an indifference relation. Analogously, there is a unique ranking based on board points. Without scale invariance, the ranking may depend on the results matrix chosen such as wins are represented by $r_{ij} = 0.5$ or $r_{ij} = 1$, for example.


Generalized row sum and least squares methods use all information of the tournament (about the opponents, opponents of opponents and so on) to break the ties. Consequently, it is very unlikely that teams remain tied, unless they have exactly the same opponents (and in such a case it seems reasonable not to break the tie). The elimination of arbitrary tie-breaking rules is a substantial advantage over official rankings.

\section{Application: European chess team championships} \label{Sec3}

In the following section, the theoretical model suggested in Section~\ref{Sec3} will be scrutinized in practice.

\subsection{Examples and implementation} \label{Sec31}

The method proposed in Section~\ref{Sec3} is illustrated with an analysis of two chess team tournaments:
\begin{itemize}[label=$\bullet$]
\item
18th European Team Chess Championship (ETCC) open tournament, 3rd-11th November 2011, Porto Carras, Greece. \\
Webpage: \url{http://euro2011.chessdom.com/} \\
Tournament rules: \citet{ECU2012} \\
Detailed results: \url{http://chess-results.com/tnr57856.aspx}
\item
19th European Team Chess Championship open tournament, 7th-18th November 2013, Warsaw, Poland. \\
Webpage: \url{http://etcc2013.com/} \\
Tournament rules: \citet{ECU2013} \\
Detailed results: \url{http://chess-results.com/tnr114411.aspx}
\end{itemize}

In both tournaments the number of competing teams was $n = 38$, playing on $2t = 4$ tables during $c = 9$ rounds. Results are known for about one quarter of possible pairs, $9 \times 19 = 171$ from $n(n-1)/2 = 703$.\footnote{~Match results can be found in Tables~A.1 (2011) and A.2 (2013), and -- in another form --  in Tables~A.3 (2011) and A.4 (2013) of \citet{Csato2016a}.}

The official ranking was based on the number of match points in both cases but tie-breaking rules were different.
In the 2013 competition application of the first tie-breaking rule (Olympiad-Sonneborn-Berger points) was enough, while in 2011 two tie-breaking rules (board points and Buchholz points -- aggregated board points of the opponents) should be used in some cases.

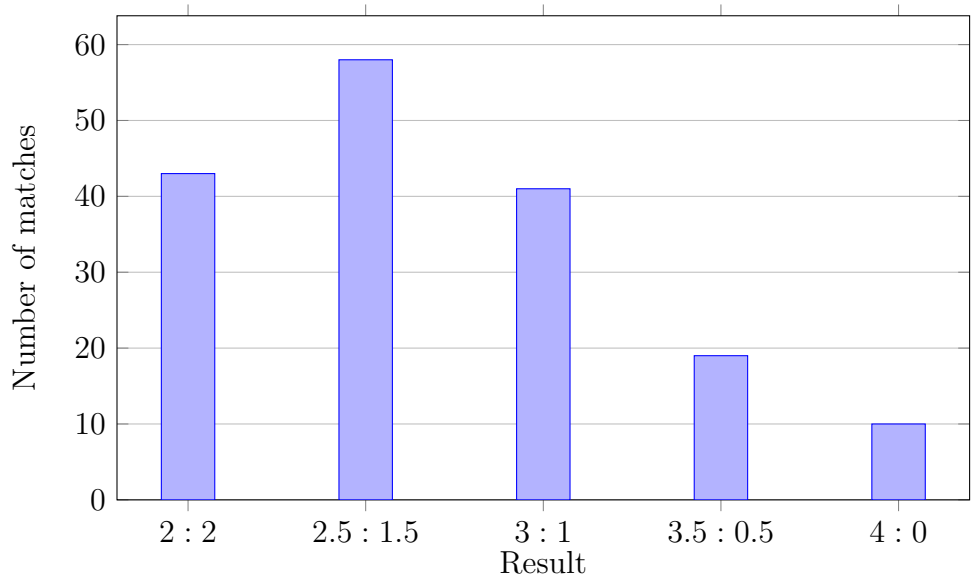
\begin{figure}[ht!]
\centering
\caption{Distribution of ($173$) match results, ETCC 2013}
\label{Fig1}

\begin{tikzpicture}
\begin{axis}[width=0.8\textwidth, 
height=0.5\textwidth,
symbolic x coords={$2:2$,{$2.5:1.5$},$3:1$,{$3.5:0.5$},$4:0$},
xtick=data, 
xlabel = Result,
ylabel = Number of matches, 
ybar,
ymin = 0,
ymajorgrids = true,
bar width=20pt]

\addplot coordinates {
($2:2$,43)
({$2.5:1.5$},58)
($3:1$,41)
({$3.5:0.5$},19)
($4:0$,10)
};
\end{axis}
\end{tikzpicture}
\end{figure}

Distribution of match results for ETCC 2013 is drawn in Figure~\ref{Fig1}. Minimal victory ($2.5:1.5$) is the mode, so incorporating board points will not influence the rankings much.

There are two exogenous rankings called \emph{Official} according to the tournament rules and \emph{Start} based on \'El\H{o} points of players, reflecting the past performance of team members. Further $12$ rankings have been calculated from the ranking problem representation. Four results matrices have been considered: $R^{MP}$, $R^{MB} = R^P(1/4) = 3/4 \, R^{MP} +  1/4 \, R^{BP}$, $R^{BM} = R^P(2/3) = 1/3 \, R^{MP} +  2/3 \, R^{BP}$ and $R^{BP}$. They were plugged into three methods, least squares ($LS$) and generalized row sum with $\varepsilon_1 = 1/324$ ($GRS_1$) and $\varepsilon_2 = 1/6$ ($GRS_2$). Note that $\varepsilon_1$ is smaller and $\varepsilon_2$ is larger than the reasonable upper bound of $1/36$.

Existence of a unique least squares solution requires connectedness of the comparison multigraph (Proposition~\ref{Prop3}), which is provided after the third round.
Rankings in the first two rounds are highly unreliable, therefore they were eliminated. From the third round all methods give one, thus $7 \times 13 + 1 = 92$ rankings will be analysed as Start remains unchanged.\footnote{~Rankings according to different methods are displayed in \citet[Tables~A.5~(2011) and A.6~(2013)]{Csato2016a}.}

Start and Official rankings are strict, that is, they do not allow for ties by definition. It can be checked that the other rankings also give a linear order of teams in all cases.

\begin{notation} \label{Not2}
The $14$ final rankings are denoted by Start, Official; $GRS_1(R^{MP})$, $GRS_1(R^{MB})$, $GRS_1(R^{BM})$, $GRS_1(R^{BP})$; $GRS_2(R^{MP})$, $GRS_2(R^{MB})$, $GRS_2(R^{BM})$, $GRS_2(R^{BP})$; and $LS(R^{MP})$, $LS(R^{MB})$, $LS(R^{BM})$, $LS(R^{BP})$. In the figures they are abbreviated by Start, Off; G1, G2, G3, G4; S1, S2, S3, S4; and L1, L2, L3, L4, respectively.
\end{notation}

\subsection{Visualisation of the rankings} \label{Sec42}

For the comparison of final rankings their distances have been calculated according to the well-known \emph{Kemeny distance} \citep{Kemeny1959} and its weighted version proposed by \citet{Can2014}. Both distances are defined on the domain of strict rankings, i.e. ties are not allowed. Our rankings satisfy this condition.

Kemeny distance is the number of pairs of alternatives ranked oppositely in the two rankings examined.

\begin{example} \label{Examp1}
The Kemeny distance of $a \succ b \succ c$ and $b \succ a \succ c$ is $1$, because they only disagree on how to order $a$ and $b$. \\
The Kemeny distance of $a \succ b \succ c$ and $a \succ c \succ b$ is $1$ because of the sole disagreement on how to order $b$ and $c$.
\end{example}

Kemeny distance was characterized by \citet{KemenySnell1962}, however, \citet{CanStorcken2013} achieved the same result without one condition. \citet{CanStorcken2013} also provides an extensive overview about the origin of this measure.

According to Example~\ref{Examp1}, the dissimilarity between $a \succ b \succ c$ and $b \succ a \succ c$ and between $a \succ b \succ c$ and $a \succ c \succ b$ by the Kemeny distance is identical. However, in our chess example a disagreement at the top of the rankings may be more significant than a disagreement at the bottom since the audience is usually interested in the first three, five or ten places but people are not bothered much whether a team is the $31$st or $34$th.

For this purpose, \citet{Can2014} proposes some functions on strict rankings in the spirit of Kemeny metric. They are respectful to the number of swaps but allow for variation in the treatment of different pairs of disagreements by weighting them according to an exogenous weight vector.
It has some price since the calculation will depend on the order of swaps between the two rankings. \citet[Theorem~1]{Can2012} shows that only the path-minimizing function satisfies the triangular inequality condition for all possible weight vectors. Finding the path-minimizing metric is not trivial, it is equivalent to solving a short-path problem in general, but the solution is known if the weights are monotonically decreasing (increasing) from the upper parts of a ranking to the lower parts.\footnote{~Then the path-minimizing metric is equivalent to winners' and losers' decomposition (the Lehmer function and the inverse Lehmer function), respectively \citep[Corollaries 1 and 2]{Can2014}.}

These results have inspired us to choose a monotonically decreasing weight vector meaning that swaps in the first places are more important than changes at the bottom of the rankings.

\begin{definition}
The weight vector of our weighted distance is given by $\omega_i = 1/i$ for all $i = 1,2, \dots,n-1$.
\end{definition}

\begin{example} \label{Examp2}
The weighted distance of $a \succ b \succ c$ and $b \succ a \succ c$ is $1$ (a swap at the first position). \\
The weighted distance of $a \succ b \succ c$ and $a \succ c \succ b$ is $1/2$ (a swap at the second position).
\end{example}

\begin{lemma}
The maximum of Kemeny distance is $n(n-1)/2 (= 703)$ and the maximum of weighted distance is $n-1 (= 37)$ if and only if the two rankings are entirely opposite.
\end{lemma}

\begin{proof}
The maximal number of swaps between two rankings is $n(n-1)/2$ in the case of two entirely opposite rankings, which is also their Kemeny distance. \\
Take the ranking $a_1 \succ a_2 \succ \dots \succ a_n$. The winners' decomposition \citep[Example~2]{Can2014} first permutes $a_n$ to the first place, which involves one swap in each position from the first to the $(n-1)$th, contributing by $1 + 1/2 + \dots 1/(n-1)$ to the weighted distance. Thereafter, it permutes $a_{n-1}$ to the second place, which involves one swap in each position from the second to the $(n-1)$th, contributing by $1/2 + \dots 1/(n-1)$ to the weighted distance, and so on.
Thus the total weighted distance of two entirely opposite rankings is $1 \times 1 + 2 \times 1/2 + 3 \times 1/3 + \dots + (n-1) \times 1/(n-1) = n-1$.
\end{proof}

We do not know about any other application of \citet{Can2014}'s novel method.

\begin{table}[!ht]
\centering
\caption{Distances of rankings, ETCC 2011}
\label{TableA7}

\begin{subtable}{\linewidth}
\caption{Kemeny distance: a swap in the $k$th position has a weight of $1$}
\label{TableA7a}
\centering
\begin{tiny}
\noindent\makebox[\textwidth]{
\rowcolors{1}{gray!25}{}
    \begin{tabularx}{1.05\linewidth}{l CCCCC CCCCC CCCC} \hline \hiderowcolors
        & \rotatebox[origin=c]{75}{Start} & \rotatebox[origin=c]{75}{Official} & \rotatebox[origin=c]{75}{$GRS_1(R^{MP})$} & \rotatebox[origin=c]{75}{$GRS_2(R^{MP})$} & \rotatebox[origin=c]{75}{$LS(R^{MP})$} & \rotatebox[origin=c]{75}{$GRS_1(R^{MB})$} & \rotatebox[origin=c]{75}{$GRS_2(R^{MB})$} & \rotatebox[origin=c]{75}{$LS(R^{MB})$} & \rotatebox[origin=c]{75}{$GRS_1(R^{BM})$} & \rotatebox[origin=c]{75}{$GRS_2(R^{BM})$} & \rotatebox[origin=c]{75}{$LS(R^{BM})$} & \rotatebox[origin=c]{75}{$GRS_1(R^{BP})$} & \rotatebox[origin=c]{75}{$GRS_2(R^{BP})$} & \rotatebox[origin=c]{75}{$LS(R^{BP})$} \\ \hline \showrowcolors
    Start & \cellcolor{gray!80}     & 107   & 100   & 98    & 100   & 107   & 99    & 96    & 110   & 93    & 93    & 130   & 99    & 85 \\
    Official & 107   & \cellcolor{gray!80}     & 37    & 45    & 73    & 0     & 38    & 69    & 25    & 34    & 60    & 71    & 52    & 60 \\ \hline
    $GRS_1(R^{MP})$  & 100   & 37    & \cellcolor{gray!80}     & 16    & 44    & 37    & 13    & 42    & 62    & 31    & 43    & 108   & 61    & 53 \\
    $GRS_2(R^{MP})$ & 98    & 45    & 16    & \cellcolor{gray!80}     & 28    & 45    & 7     & 26    & 70    & 27    & 29    & 114   & 67    & 45 \\
    $LS(R^{MP})$ & 100   & 73    & 44    & 28    & \cellcolor{gray!80}     & 73    & 35    & 8     & 94    & 47    & 21    & 130   & 81    & 41 \\ \hline
    $GRS_1(R^{MB})$ & 107   & 0     & 37    & 45    & 73    & \cellcolor{gray!80}     & 38    & 69    & 25    & 34    & 60    & 71    & 52    & 60 \\
    $GRS_2(R^{MB})$ & 99    & 38    & 13    & 7     & 35    & 38    & \cellcolor{gray!80}     & 33    & 63    & 20    & 32    & 107   & 60    & 40 \\
    $LS(R^{MB})$ & 96    & 69    & 42    & 26    & 8     & 69    & 33    & \cellcolor{gray!80}     & 88    & 41    & 13    & 122   & 73    & 33 \\ \hline
    $GRS_1(R^{BM})$ & 110   & 25    & 62    & 70    & 94    & 25    & 63    & 88    & \cellcolor{gray!80}     & 49    & 79    & 46    & 41    & 71 \\
    $GRS_2(R^{BM})$ & 93    & 34    & 31    & 27    & 47    & 34    & 20    & 41    & 49    & \cellcolor{gray!80}     & 30    & 87    & 40    & 26 \\
    $LS(R^{BM})$ & 93    & 60    & 43    & 29    & 21    & 60    & 32    & 13    & 79    & 30    & \cellcolor{gray!80}     & 111   & 60    & 20 \\ \hline
    $GRS_1(R^{BP})$ & 130   & 71    & 108   & 114   & 130   & 71    & 107   & 122   & 46    & 87    & 111   & \cellcolor{gray!80}     & 57    & 97 \\
    $GRS_2(R^{BP})$ & 99    & 52    & 61    & 67    & 81    & 52    & 60    & 73    & 41    & 40    & 60    & 57    & \cellcolor{gray!80}     & 44 \\
    $LS(R^{BP})$ & 85    & 60    & 53    & 45    & 41    & 60    & 40    & 33    & 71    & 26    & 20    & 97    & 44    & \cellcolor{gray!80} \\ \hline
    \end{tabularx} }
\end{tiny}
\end{subtable}
\vspace{0.5cm}

\begin{subtable}{\linewidth}
\caption{Weighted distance: a swap in the $k$th position has a weight of $1/k$}
\label{TableA7b}
\centering
\begin{tiny}
\noindent\makebox[\textwidth]{
\rowcolors{1}{gray!25}{}
    \begin{tabularx}{1.05\linewidth}{l CCCCC CCCCC CCCC} \hline \hiderowcolors
        & \rotatebox[origin=c]{75}{Start} & \rotatebox[origin=c]{75}{Official} & \rotatebox[origin=c]{75}{$GRS_1(R^{MP})$} & \rotatebox[origin=c]{75}{$GRS_2(R^{MP})$} & \rotatebox[origin=c]{75}{$LS(R^{MP})$} & \rotatebox[origin=c]{75}{$GRS_1(R^{MB})$} & \rotatebox[origin=c]{75}{$GRS_2(R^{MB})$} & \rotatebox[origin=c]{75}{$LS(R^{MB})$} & \rotatebox[origin=c]{75}{$GRS_1(R^{BM})$} & \rotatebox[origin=c]{75}{$GRS_2(R^{BM})$} & \rotatebox[origin=c]{75}{$LS(R^{BM})$} & \rotatebox[origin=c]{75}{$GRS_1(R^{BP})$} & \rotatebox[origin=c]{75}{$GRS_2(R^{BP})$} & \rotatebox[origin=c]{75}{$LS(R^{BP})$} \\ \hline \showrowcolors
    Start & \cellcolor{gray!80}     & 10.79 & 9.68  & 9.60  & 9.39  & 10.79 & 9.54  & 9.15  & 11.30 & 9.45  & 8.66  & 12.16 & 10.05 & 8.10 \\
    Official & 10.79 & \cellcolor{gray!80}     & 3.04  & 3.67  & 6.33  & 0.00  & 3.08  & 6.12  & 2.09  & 2.74  & 5.62  & 6.55  & 4.89  & 5.58 \\ \hline
    $GRS_1(R^{MP})$ & 9.68  & 3.04  & \cellcolor{gray!80}     & 1.02  & 3.80  & 3.04  & 0.75  & 3.73  & 5.04  & 2.29  & 3.80  & 9.41  & 5.87  & 4.39 \\
    $GRS_2(R^{MP})$ & 9.60  & 3.67  & 1.02  & \cellcolor{gray!80}     & 2.80  & 3.67  & 0.60  & 2.73  & 5.66  & 2.24  & 2.87  & 9.97  & 6.36  & 4.15 \\
    $LS(R^{MP})$ & 9.39  & 6.33  & 3.80  & 2.80  & \cellcolor{gray!80}     & 6.33  & 3.39  & 0.53  & 8.07  & 4.36  & 1.53  & 9.94  & 6.20  & 3.01 \\ \hline
    $GRS_1(R^{MB})$ & 10.79 & 0.00  & 3.04  & 3.67  & 6.33  & \cellcolor{gray!80}     & 3.08  & 6.12  & 2.09  & 2.74  & 5.62  & 6.55  & 4.89  & 5.58 \\
    $GRS_2(R^{MB})$ & 9.54  & 3.08  & 0.75  & 0.60  & 3.39  & 3.08  & \cellcolor{gray!80}     & 3.31  & 5.09  & 1.65  & 3.27  & 9.42  & 5.80  & 3.69 \\
    $LS(R^{MB})$ & 9.15  & 6.12  & 3.73  & 2.73  & 0.53  & 6.12  & 3.31  & \cellcolor{gray!80}     & 7.74  & 4.04  & 1.00  & 9.48  & 5.71  & 2.49 \\ \hline
    $GRS_1(R^{BM})$ & 11.30 & 2.09  & 5.04  & 5.66  & 8.07  & 2.09  & 5.09  & 7.74  & \cellcolor{gray!80}     & 4.10  & 7.20  & 4.48  & 3.96  & 6.58 \\
    $GRS_2(R^{BM})$ & 9.45  & 2.74  & 2.29  & 2.24  & 4.36  & 2.74  & 1.65  & 4.04  & 4.10  & \cellcolor{gray!80}     & 3.29  & 8.01  & 4.23  & 2.98 \\
    $LS(R^{BM})$ & 8.66  & 5.62  & 3.80  & 2.87  & 1.53  & 5.62  & 3.27  & 1.00  & 7.20  & 3.29  & \cellcolor{gray!80}     & 8.79  & 4.79  & 1.49 \\ \hline
    $GRS_1(R^{BP})$ & 12.16 & 6.55  & 9.41  & 9.97  & 9.94  & 6.55  & 9.42  & 9.48  & 4.48  & 8.01  & 8.79  & \cellcolor{gray!80}     & 4.53  & 7.78 \\
    $GRS_2(R^{BP})$ & 10.05 & 4.89  & 5.87  & 6.36  & 6.20  & 4.89  & 5.80  & 5.71  & 3.96  & 4.23  & 4.79  & 4.53  &  \cellcolor{gray!80} & 3.64 \\
    $LS(R^{BP})$ & 8.10  & 5.58  & 4.39  & 4.15  & 3.01  & 5.58  & 3.69  & 2.49  & 6.58  & 2.98  & 1.49  & 7.78  & 3.64  & \cellcolor{gray!80} \\ \hline
    \end{tabularx} }
\end{tiny}
\end{subtable}
\end{table}

Distances of rankings of ETCC 2011 is presented in Table~\ref{TableA7}. All Kemeny distances are significantly smaller than its maximum of $703$ for entirely opposite rankings. Largest values usually occur in comparison with Start since the latter is not influenced by the results. However, rankings based on match points and board points are also relatively far from each other. Official coincides with $GRS_1(R^{MB})$.

Weighted distances are presented in Table~\ref{TableA7b}. Its maximum is $37$. Ratio of Kemeny and weighted distances are between $8.73$ and $17.44$ for ETCC 2011, and between $5.81$ and $18.73$ for ETCC 2013. In the second case accounting for swaps' positions has a larger effect but the discrepancy of the two distances remains smaller than expected, that is, variations are more or less equally distributed along the rankings. 

The ranking from $GRS_1(R^{MP})$ means a kind of tie-breaking rule for match points both in ETCC 2011 and ETCC 2013: generalized row sum gives the match points ranking for $\varepsilon = 0$, while a small increase in the parameter breaks ties among teams according to the strength of their opponents. The official ranking also aims to eliminate ties, although it uses a different approach.

The pairwise distances of $14$ rankings can be plotted in a $13$-dimensional space without loss of information but it still seems to be unmanageable. Therefore multidimensional scaling \citep{KruskalWish1978} has been applied, similarly to \citet{Csato2013a}. It is a statistical method in information visualization for exploring similarities or dissimilarities in data, a textbook application of MDS is to draw cities on a map from the matrix consisting of their air distances.

Kemeny and weighted distances are measured on a ratio scale due to the existence of a natural minimum and maximum. Then discrepancies of the reduced dimensional map are linear functions of the original distances.
Both Stress and RSQ tests for validity strengthen that two dimensions are sufficient to plot the $14$ rankings, however, one dimension is too restrictive. The method produces a map where only the position of objects count, more similar rankings are closer to each other. Only the distances of points representing the rankings yield information, the meaning of the axes remains obscure.

MDS maps reinforce the conjecture from Table~\ref{TableA7} that Start is far away from all other rankings (see  \citet[Figure 2]{Csato2016a}). Thus Start ranking is omitted from further analysis, which improves the mapping, too.

\begin{figure}[!ht]
\centering
\caption{MDS maps of rankings, ETCC 2013}
\label{Fig3}
  
\begin{subfigure}{\textwidth}
  \centering
  \caption{Kemeny distance, without Start}
  \label{Fig3a}
\begin{tikzpicture}[scale=2.5, auto=center, transform shape]
\tikzstyle{every node}=[font=\tiny]
\draw (0.0138,-0.4964) node {$\bullet$};
\node [below] at (0.0138,-0.4964) {Off};

\draw (0.534,-0.5716) node {\ding{64}};
\node [below] at (0.534,-0.5716) {G1};

\draw (0.9647,-0.3498) node {$\times$};
\node [below] at (0.9647,-0.3498) {S1};

\draw (1.6836,0.145) node {$\diamond$};
\node [below] at (1.6836,0.145) {L1};

\draw (-0.8427,-0.7025) node {\ding{64}};
\node [below] at (-0.8427,-0.7025) {G2};

\draw (0.6915,-0.2791) node {$\times$};
\node [above] at (0.6915,-0.2791) {S2};

\draw (1.4897,0.2395) node {$\diamond$};
\node [above] at (1.4897,0.2395) {L2};

\draw (-1.4181,-0.6437) node {\ding{64}};
\node [below] at (-1.4181,-0.6437) {G3};

\draw (0.1302,0.006) node {$\times$};
\node [below] at (0.1302,0.006) {S3};

\draw (1.0336,0.5183) node {$\diamond$};
\node [below] at (1.0336,0.5183) {L3};

\draw (-3.1539,0.4271) node {\ding{64}};
\node [below] at (-3.1539,0.4271) {G4};

\draw (-1.3134,0.666) node {$\times$};
\node [below] at (-1.3134,0.666) {S4};

\draw (0.1869,1.0412) node {$\diamond$};
\node [below] at (0.1869,1.0412) {L4};
\end{tikzpicture}
\end{subfigure}
\vspace{1cm}

\begin{subfigure}{\textwidth}
  \centering
  \caption{Weighted distance, without Start}
  \label{Fig3b}

\begin{tikzpicture}[scale=2.5, auto=center, transform shape]
\tikzstyle{every node}=[font=\tiny]
\draw (0.2885,-0.4703) node {$\bullet$};
\node [below] at (0.2885,-0.4703) {Off};

\draw (0.7103,-0.4816) node {\ding{64}};
\node [below] at (0.7103,-0.4816) {G1};

\draw (0.9486,-0.3615) node {$\times$};
\node [right] at (0.9486,-0.3615) {S1};

\draw (1.4247,0.4121) node {$\diamond$};
\node [below] at (1.4247,0.4121) {L1};

\draw (-0.3476,-0.789) node {\ding{64}};
\node [below] at (-0.3476,-0.789) {G2};

\draw (0.7663,-0.3182) node {$\times$};
\node [above] at (0.7663,-0.3182) {S2};

\draw (1.2999,0.4074) node {$\diamond$};
\node [above] at (1.2999,0.4074) {L2};

\draw (-1.3262,-0.6883) node {\ding{64}};
\node [below] at (-1.3262,-0.6883) {G3};

\draw (0.1168,-0.3467) node {$\times$};
\node [above] at (0.1168,-0.3467) {S3};

\draw (0.9023,0.6299) node {$\diamond$};
\node [below] at (0.9023,0.6299) {L3};

\draw (-3.2791,0.3743) node {\ding{64}};
\node [below] at (-3.2791,0.3743) {G4};

\draw (-1.5192,0.4081) node {$\times$};
\node [below] at (-1.5192,0.4081) {S4};

\draw (0.0147,1.2237) node {$\diamond$};
\node [below] at (0.0147,1.2237) {L4};
\end{tikzpicture}
\end{subfigure}
\end{figure}

There is not much difference between the four charts (ETCC 2011 vs. 2013, Kemeny vs. weighted distances). MDS procedures of ETCC 2013 and Kemeny distances have more favourable validity measures than MDS procedures of ETCC 2011 and weighted distances. They reveal the following results shown by Figure~\ref{Fig3}:
\begin{enumerate}
\item
Start significantly differs from all other rankings since it does not depend on the results of the tournament;
\item
Generalized row sum rankings (with low $\lambda$) are more similar to the official one than least squares;
\item
The order of results matrices by the variability of rankings for a given scoring method is $R^{MP} < R^{MB} < R^{BM} < R^{BP}$, a greater role of match points (smaller $\lambda$) stabilizes the rankings;
\item
The order of scoring procedures by the variability of rankings for a given results matrix is $LS < GRS_2 < GRS_1$, a greater influence of opponents' results stabilizes the rankings;
\item
The effect of tie-breaking rule for match points is not negligible (Off and G1 are not very close to each other).
\end{enumerate}

On the basis of these observations, the application of least squares with a generalized results matrix favouring match points (a low $\lambda$, for example, $1/4$ as in $R^{MB}$) is proposed for ranking in Swiss system chess team tournaments. It gives an incentive to score more board points but still prefers match points.

\subsection{Analysis of a ranking} \label{Sec43}

The decomposition of the least squares rating \citep{Csato2015a} offers another approach to compare the rankings. The ranking problem is balanced and the comparison multigraph is regular, therefore Proposition~\ref{Prop4} can be applied. In the zeroth step ($\mathbf{q}^{(0)}$) it gives the match points ranking, the official ranking without the application of tie-breaking rules. After that, the iterated ratings reflect the strength of opponents, opponents of opponents and so on by accounting for their average match points as $\mathfrak{d}I-L = M$.
A ranking equivalent to $\mathbf{q}(R^{MP})$ is obtained after the seventh (from $\mathbf{q}^{(7)}(R^{MP})$) and after the twelfth step (from $\mathbf{q}^{(12)}(R^{MP})$) in the case of ETCC 2011 and ETCC 2013, respectively.

\begin{table}[!ht]
\centering
\caption[Positional changes in the decomposition of $LS(R^{MP})$, ETCC 2013]{Positional changes in the decomposition of $LS(R^{MP})$, ETCC 2013 \\
\vspace{0.25cm}
\footnotesize{Rank improvements and declines between the rankings from the corresponding $\mathbf{q}^{(k)}$ rating are indicated by the arrows $\up$ and $\down$, respectively. A number in brackets represent the same number of $\up$ or $\down$ arrows. Lack of change is indicated by --.}
}
\label{Table1}
\begin{footnotesize}
\rowcolors{1}{gray!25}{}
    \begin{tabularx}{\textwidth}{L ccccc cccccc C} \toprule \hiderowcolors
    \multirow{2}{\linewidth}{Team} & \multicolumn{11}{c}{Value of $k$ in $\mathbf{q}^{(k)}$} \\
     & Off (0) & 1     & 2     & 3     & 4     & 5     & 6     & 9     & 12    & Cumulated & LS ($\infty$) \\ \midrule \showrowcolors
    Azerbaijan & 1     & --    & \down & --    & --    & --    & --    & --    & --    & \down & 2 \\
    France & 2     & --    & \up & --    & --    & --    & --    & --    & --    & \up & 1 \\
    Russia & 3     & \down & --    & --    & --    & --    & --    & --    & --    & \down & 4 \\
    Armenia & 4     & \up & --    & --    & --    & --    & --    & --    & --    & \up & 3 \\
    Hungary & 5     & --    & --    & --    & --    & --    & --    & --    & --    & --    & 5 \\
    Georgia & 6     & --    & --    & --    & --    & --    & --    & --    & --    & --    & 6 \\
    Greece & 7     & --    & --    & --    & --    & --    & \down & --    & --    & \down & 8 \\
    Czech Rep. & 8     & \down & \down & --    & --    & --    & --    & --    & --    & \down \down & 10 \\
    Ukraine & 9     & \up & --    & --    & --    & --    & \up & --    & --    & \up \up & 7 \\
    England & 10    & --    & \up & --    & --    & --    & --    & --    & --    & \up & 9 \\
    Netherlands & 11    & \down \, (6) & --    & --    & --    & --    & --    & --    & --    & \down \, (6) & 17 \\
    Italy & 12    & \up & --    & --    & --    & \down & --    & --    & --    & --    & 12 \\
    Serbia & 13    & \down \down \down & \down \down & --    & --    & \down & --    & --    & --    & \down \, (6) & 19 \\
    Romania & 14    & \down \, (4) & \up \up & --    & \up & --    & --    & --    & --    & \down & 15 \\
    Belarus & 15    & \up \up \up & --    & --    & --    & \up & --    & --    & --    & \up \, (4) & 11 \\
    Poland & 16    & \up \up \up & --    & --    & --    & \down & --    & --    & --    & \up \up & 14 \\
    Croatia & 17    & \up \up & --    & --    & \down & --    & --    & --    & --    & \up & 16 \\
    Montenegro & 18    & \down & --    & \down & --    & --    & --    & --    & \down & \down \down \down & 21 \\
    Spain & 19    & \down \down & --    & --    & --    & --    & \down & --    & --    & \down \down \down & 22 \\
    Germany & 20    & --    & --    & \up & --    & \up & --    & --    & --    & \up \up & 18 \\
    Slovenia & 21    & \up \, (7) & --    & --    & --    & \up & --    & --    & --    & \up \, (8) & 13 \\
    Poland Futures & 22    & \down \down & --    & \down & --    & \down & --    & --    & --    & \down \, (4) & 26 \\
    Lithuania & 23    & \down \down & \down \, (4) & --    & --    & --    & \down & --    & --    & \down \, (7) & 30 \\
    Turkey & 24    & \up \up & --    & --    & --    & --    & \up & --    & \up & \up \, (4) & 20 \\
    Bulgaria & 25    & \up \up & --    & --    & --    & --    & --    & --    & --    & \up \up & 23 \\
    Sweden & 26    & \down & --    & \down & --    & --    & --    & --    & --    & \down \down & 28 \\
    Denmark & 27    & \down \down \down & \down & --    & --    & --    & --    & \down & --    & \down \, (5) & 32 \\
    Israel & 28    & \up \up & \up & \up & --    & --    & --    & --    & --    & \up \, (4) & 24 \\
    Iceland & 29    & \down \down \down & --    & --    & --    & --    & --    & \up & --    & \down \down & 31 \\
    Austria & 30    & \up \up & \up \up & --    & --    & \up & --    & --    & --    & \up \, (5) & 25 \\
    Poland Goldies & 31    & --    & \up & --    & --    & --    & \up & --    & --    & \up \up & 29 \\
    Switzerland & 32    & \up \up \up & \up & \up & --    & --    & --    & --    & --    & \up \, (5) & 27 \\
    Belgium & 33    & --    & --    & \down & --    & --    & --    & --    & --    & \down & 34 \\
    Finland & 34    & --    & --    & \up & --    & --    & --    & --    & --    & \up & 33 \\
    Norway & 35    & --    & --    & --    & --    & --    & --    & --    & --    & --    & 35 \\
    Scotland & 36    & --    & --    & --    & --    & --    & --    & --    & --    & --    & 36 \\
    FYR Macedonia & 37    & --    & --    & --    & --    & --    & --    & --    & --    & --    & 37 \\
    Wales & 38    & --    & --    & --    & --    & --    & --    & --    & --    & --    & 38 \\ \hline
    \end{tabularx} 
\end{footnotesize}
\end{table}

Table~\ref{Table1} shows the changes of teams' positions in each step of the decomposition of the ranking $LS(R^{MP})$ for ETCC 2013. In the second column ($\mathbf{q}^{(0)}$), ties are broken according to the official rules, so it coincides with the official ranking. In subsequent steps there are no ties.
The last change is a swap of Turkey and Montenegro in the twelfth step of the iteration. The least squares method is far from being only a tie-breaking rule for match points (contrary to generalized row sum with $\varepsilon_1 = 1/324$), a team may overtake another one despite its disadvantage of two match points.

Correction according to opponents' strength results in seven positions improvement for Slovenia together with a four positions decline for Romania and six for Netherlands. Hence Slovenia overtakes Netherlands despite it has a two match points disadvantage.\footnote{~Tie-breaking rule $TB4$ (aggregated number of board points of the opponents) shows a similar direction of adjustment.}
Subsequent steps of the iteration usually result in a similar direction of swaps, however, in a more moderated extent. A notable exception is Romania, regaining some positions due to indirect opponents. The monotonic decrease of absolute adjustments is violated only by Lithuania.

There are two changes among the top six teams. France becomes the winner of the tournament after $k=2$ instead of Azerbaijan. It can be debated since the latter team has no loss, however, the schedule of France was more difficult. The swap of Russia and Armenia may be explained by the advance on an outer circle of the former team (i.e. Russia had a worse performance than Armenia during the first rounds of the tournament). 

\begin{table}[!ht]
\centering
\caption[Positional changes in the decomposition of $LS(R^{MP})$, ETCC 2011]{Positional changes in the decomposition of $LS(R^{MP})$, ETCC 2011 \\
\vspace{0.25cm}
\footnotesize{Rank improvements and declines between the rankings from the corresponding $\mathbf{q}^{(k)}$ rating are indicated by the arrows $\up$ and $\down$, respectively. A number in brackets represents the same number of $\up$ or $\down$ arrows. Lack of change is indicated by --.}
}
\label{TableA8}
\begin{footnotesize}
\rowcolors{1}{gray!25}{}
    \begin{tabularx}{\textwidth}{L ccccc cccccc C} \toprule \hiderowcolors
    \multirow{2}{*}{Team} & \multicolumn{10}{c}{Value of $k$ in $\mathbf{q}^{(k)}$} \\
    & Off (0) & 1     & 2     & 3     & 4     & 5     & 7     & 8     & Cumulated & LS ($\infty$) \\ \midrule \showrowcolors
    Germany & 1     & --    & --    & --    & --    & \down & --    & --    & \down & 2 \\
    Azerbaijan & 2     & --    & --    & --    & --    & \up   & --    & --    & \up   & 1 \\
    Hungary & 3     & \down \down \down & --    & --    & --    & --    & --    & --    & \down \down \down & 6 \\
    Armenia & 4     & \down & --    & --    & --    & --    & --    & --    & \down & 5 \\
    Russia & 5     & \up \up & --    & --    & --    & --    & --    & --    & \up \up & 3 \\
    Netherlands & 6     & \down & --    & \down & --    & --    & --    & --    & \down \down & 8 \\
    Bulgaria & 7     & \up \up \up & --    & --    & --    & --    & --    & --    & \up \up \up & 4 \\
    Poland & 8     & \down \, (6) & --    & \down & --    & --    & \down & --    & \down \, (8) & 16 \\
    Romania & 9     & --    & \down & \down & --    & \down & --    & --    & \down \down \down & 12 \\
    Spain & 10    & \up \up & --    & \up   & --    & --    & --    & --    & \up \up \up & 7 \\
    Italy & 11    & \up   & \up   & --    & --    & --    & --    & --    & \up \up & 9 \\
    Serbia & 12    & \down \, (7) & --    & --    & \down \down & --    & --    & --    & \down \, (9) & 21 \\
    Georgia & 13    & \down \, (9) & \down & \down & \down & \down & --    & \down & \down \, (14) & 27 \\
    Israel & 14    & \down \down & --    & --    & --    & --    & \up   & --    & \down & 15 \\
    Ukraine & 15    & \up \up \up & \up   & \up   & --    & --    & \down & --    & \up \, (4) & 11 \\
    Czech Rep. & 16    & \up \up \up & \down \down & \up   & --    & --    & --    & --    & \up \up & 14 \\
    Slovenia & 17    & \up \, (6) & \down \down & --    & --    & --    & --    & --    & \up \, (4) & 13 \\
    Moldova & 18    & \down \down & --    & --    & --    & --    & --    & --    & \down \down & 20 \\
    France & 19    & \up \, (4) & \up \up \up & --    & --    & \up   & \up   & --    & \up \, (9) & 10 \\
    Greece & 20    & \up \up \up & --    & --    & --    & --    & --    & --    & \up \up \up & 17 \\
    Croatia & 21    & \up \up \up & --    & --    & --    & --    & --    & --    & \up \up \up & 18 \\
    England & 22    & \up   & --    & --    & \up \up & --    & --    & --    & \up \up \up & 19 \\
    Switzerland & 23    & \down \, (4) & \up \up \up & \up   & --    & --    & --    & --    & --    & 23 \\
    Latvia & 24    & \up   & \up   & --    & --    & --    & --    & --    & \up \up & 22 \\
    Montenegro & 25    & --    & \down & \down \down \down & --    & --    & --    & --    & \down \, (4) & 29 \\
    Iceland & 26    & --    & \down \down & --    & --    & --    & --    & --    & \down \down & 28 \\
    Sweden & 27    & \up \up \up & \down & --    & \up   & --    & \down & --    & \up \up & 25 \\
    Denmark & 28    & --    & \up   & \up   & \down & --    & --    & \up   & \up \up & 26 \\
    Norway & 29    & \down & \down & --    & --    & --    & --    & --    & \down \down & 31 \\
    FYROM & 30    & \down \down \down & --    & --    & --    & --    & --    & --    & \down \down \down & 33 \\
    Finland & 31    & \down & --    & --    & --    & --    & --    & --    & \down & 32 \\
    Austria & 32    & \up   & \up   & --    & --    & --    & --    & --    & \up \up & 30 \\
    Lithuania & 33    & \up \, (4) & --    & \up \up & \up   & \up   & \up   & --    & \up \, (9) & 24 \\
    Turkey & 34    & --    & --    & --    & --    & --    & --    & --    & --    & 34 \\
    Scotland & 35    & --    & --    & --    & --    & --    & --    & --    & --    & 35 \\
    Luxembourg & 36    & --    & --    & --    & --    & --    & --    & --    & --    & 36 \\
    Wales & 37    & --    & --    & --    & --    & --    & --    & --    & --    & 37 \\
    Cyprus & 38    & --    & --    & --    & --    & --    & --    & --    & --    & 38 \\ \hline
    \end{tabularx} 
\end{footnotesize}
\end{table}

Imperfection of the official ranking is further highlighted by ETCC 2011, for which Table~\ref{TableA8} contains the positional changes according to the iterative decomposition of $LS(R^{MP})$.
Here France scored three wins and three draws in the first six rounds but it has been defeated three times after that, presenting an extreme example of advance on an inner circle. Thus France had a more challenging schedule compared to teams with the same number of match points, reflected in the significant adjustment by the least squares method.

On the other side, Serbia loses nine, and Georgia loses $14$ positions. They had luck with the opponents, for example, Georgia had not played against a better team according to the official ranking, which is quite strange for a team at the $13$th place. Consequently, both Serbia and Georgia significantly benefit from decreasing $\varepsilon$ or increasing the role of board points.    

\subsection{Assessment of the rankings} \label{Sec44}

The $14$ rankings are evaluated from three aspects:
\begin{itemize}[label=$\bullet$]
\item
Predictive performance: ability to forecast the outcomes of future matches;
\item
Retrodictive performance: ability to match the results of contests already played;
\item
Robustness: stability between subsequent rounds.
\end{itemize}
The first two are standard aspects for the classification of mathematical ranking models \citep{Pasteur2010}.
However, for the ranking of a Swiss system tournament, the second is much more important: the aim is to get a meaningful ranking on the basis of matches already played, shown by in-sample fit.

The third measure, stability, seems to be important because of (at least) two causes. First, both the participants and the audience may dislike if the rankings are volatile. Naturally, extreme stability is not favourable, too, but it is usually not a problem in a Swiss system tournament. The second argument for robustness may be that the number of rounds is often determined arbitrarily, for instance, it was $13$ in the 2006 and $11$ in the 2013 chess olympiads with $148$ and $146$ teams, respectively.

Predictive and retrodictive performances are measured by the number of match and board points scored by an underdog against a better team. It does not take into account the difference of positions, only its sign.

Prediction power has a meaning only after the third round, when the comparison multigraph becomes connected. Start has the most favourable forecasting performance for the remaining matches, especially in the first rounds, that is, match outcomes are determined by teams' ability rather than by their results in the competition \citep[Figure~A.1]{Csato2016a}. There is also no difference among the methods in prediction power if only the next round is scrutinized \citep[Figure~A.2]{Csato2016a}.

The fact that Start ranking is the best for forecasting match results reflects the insignificance of prediction precision for a Swiss system tournament ranking: after all, what is the meaning to organize a contest if its final result is determined by teams' ability?

Retrodictive performance has a meaning after the third round, too, however, it is also defined after the last round when prediction power cannot be interpreted.
Least squares method seems to be the best from this point of view, despite its statistical significance remains dubious \citep[Figure~A.3]{Csato2016a}. Generalized row sum is placed between the least squares and official rankings. Choice of the results matrix and the tournament does not influence these findings.
 
\begin{figure}[htbp]
\centering
\caption{Robustness (distance between subsequent rounds), ETCC 2011}
\label{Fig4}

\begin{subfigure}{\textwidth}
  \centering
  \caption{Kemeny distance, results matrix $R^{MP}$}
  \label{Fig4a}
  
\begin{tikzpicture}
\begin{axis}[width=\textwidth, 
height=0.65\textwidth,
symbolic x coords={$3-4$,$4-5$,$5-6$,$6-7$,$7-8$,$8-9$},
xtick=data,
legend entries={Off$\quad$,G1$\quad$,S1$\quad$,L1},
legend style={at={(0.5,-0.1)},anchor = north,legend columns = 6},
ymajorgrids = true,
]

\addplot[red,dashed,mark=o,mark size=3pt,mark options={solid},thick] coordinates{
($3-4$,98)
($4-5$,89)
($5-6$,61)
($6-7$,75)
($7-8$,71)
($8-9$,69)
}
node [pos=1,pin={[pin edge={white}, pin distance=0cm] 0:{Off}}] {};

\addplot[green,mark=triangle,mark size=3pt,mark options={solid},thick] coordinates{
($3-4$,100)
($4-5$,87)
($5-6$,63)
($6-7$,79)
($7-8$,53)
($8-9$,55)
}
node [pos=1,pin={[pin edge={white}, pin distance=0cm] 0:{G1}}] {};

\addplot[black,dotted,mark=diamond,mark size=3pt,mark options={solid},thick] coordinates{
($3-4$,100)
($4-5$,89)
($5-6$,57)
($6-7$,65)
($7-8$,39)
($8-9$,45)
}
node [pos=1,pin={[pin edge={white}, pin distance=0cm] 0:{S1}}] {};

\addplot[blue,dashdotted,mark=star,mark size=3pt,mark options={solid},thick] coordinates{
($3-4$,82)
($4-5$,75)
($5-6$,36)
($6-7$,44)
($7-8$,17)
($8-9$,20)
}
node [pos=1,pin={[pin edge={white}, pin distance=0cm] 0:{L1}}] {};
\end{axis}
\end{tikzpicture}
\end{subfigure}
\vspace{0.5cm}

\begin{subfigure}{\textwidth}
  \centering
  \caption{Weighted distance, results matrix $R^{MP}$}
  \label{Fig4b}
  
\begin{tikzpicture}
\begin{axis}[width=\textwidth, 
height=0.65\textwidth,
symbolic x coords={$3-4$,$4-5$,$5-6$,$6-7$,$7-8$,$8-9$},
xtick=data,
legend entries={Off$\quad$,G1$\quad$,S1$\quad$,L1},
legend style={at={(0.5,-0.1)},anchor = north,legend columns = 6},
ymajorgrids = true,
]

\addplot[red,dashed,mark=o,mark size=3pt,mark options={solid},thick] coordinates{
($3-4$, 9.0083)
($4-5$, 8.7677)
($5-6$, 5.6199)
($6-7$, 7.0279)
($7-8$, 5.9262)
($8-9$, 6.3287)
}
node [pos=1,pin={[pin edge={white}, pin distance=0cm] 0:{Off}}] {};

\addplot[green,mark=triangle,mark size=3pt,mark options={solid},thick] coordinates{
($3-4$, 6.7867)
($4-5$, 7.569)
($5-6$, 5.9343)
($6-7$, 6.4027)
($7-8$, 4.8679)
($8-9$, 4.3984)
}
node [pos=1,pin={[pin edge={white}, pin distance=0cm] 0:{G1}}] {};

\addplot[black,dotted,mark=diamond,mark size=3pt,mark options={solid},thick] coordinates{
($3-4$, 6.7956)
($4-5$, 7.5315)
($5-6$, 5.5539)
($6-7$, 5.5205)
($7-8$, 3.5767)
($8-9$, 3.3866)
}
node [pos=1,pin={[pin edge={white}, pin distance=0cm] 0:{S1}}] {};

\addplot[blue,dashdotted,mark=star,mark size=3pt,mark options={solid},thick] coordinates{
($3-4$, 5.8299)
($4-5$, 5.763)
($5-6$, 2.5053)
($6-7$, 3.6458)
($7-8$, 1.6695)
($8-9$, 1.6821)
}
node [pos=1,pin={[pin edge={white}, pin distance=0cm] 0:{L1}}] {};
\end{axis}
\end{tikzpicture}
\end{subfigure}
\end{figure}

Stability is defined as the distance of rankings in subsequent rounds. It has no meaning for Start but can be calculated for all other rankings from the third round.
Figure~\ref{Fig4} illustrates the robustness of some rankings in ETCC 2011. Variability does not decrease monotonically, but a solid decline is observed as the actual round gives relatively fewer and fewer information. Ranking $LS(R^{MP})$ is the most robust according to both Kemeny and weighted distances, followed by $GRS_2(R^{MP})$, then $GRS_1(R^{MP})$ and Official: rankings become less volatile by taking into account the performance of opponents. Difference of absolute values seems to be more significant in the case of weighted distance, the least squares method is robust especially in the first, critical places. 
The order of variability $LS < GRS_2 < GRS_1$ is valid for all other result matrices, however, $GRS_1$ is sometimes more volatile than the official ranking.

In the case of ETCC 2013, these conclusions are more uncertain but least squares remains the most stable with the exception of first rounds \citep[Figure~A.4]{Csato2016a}. Readers interested in a somewhat more detailed analysis of the two tournaments are encouraged to study \citet{Csato2016a}.

To summarize, the least squares method gives the most robust and legitimate ranking. Therefore, its application is also recommended if the organizers want to mitigate the effects of the (predetermined) number of rounds on the ranking.

\section{Discussion} \label{Sec5}

The paper has given an axiomatic analysis of ranking in Swiss system chess team tournaments. The framework is flexible with respect to the role of the opponents (parameter $\varepsilon$) and the influence of match and board points (choice of the results matrix).
The suggested methods are close to the concept of official rankings (they coincide in the case of round-robin tournaments), can be calculated iteratively or by solving a system of linear equations and have a clear interpretation on the comparison multigraph. They also do not call for arbitrary tie-breaking rules.

The model is tested on the results of the 2011 and 2013 European Team Chess Championship open tournaments, which supports the application of least squares method due to its relative insensitivity to the choice between match and board points, retrodictive accuracy and stability. There is an opportunity to take into account the number of board points scored by using a generalized results matrix favouring match points (small $\lambda$ close to zero).
The findings confirm that the official rankings have significant failures, therefore recursive methods, similar to generalized row sum and least squares, are worth to consider for ranking purposes.

Naturally, the framework may have some disadvantages \citep{Brozos-Vazquezetal2010}: a computer is needed in order to calculate the ranking of the tournament, and it will be difficult for the players to verify and understand the whole procedure. However, we agree with \citet{Forlano2011} that '\emph{The fact that players are not able to foresee the final standing should not be considered a disadvantage but a way to force the players to play each round as the decisive one.}' as well as '\emph{The fact that the the players cannot replicate the method manually should be seen of no significance.}'
While the least squares method is more complicated than usual tie-breaking rules, its simple graph interpretation \citep{Csato2015a} and its similarity to an 'infinite Buchholz' may help in the understanding.

Anyway, there usually exists a trade-off between simplicity and other favourable properties (sample fit, robustness), and the use of more developed methods is worth to consider in the case of Swiss system tournaments in order to avoid anomalies of the ranking,\footnote{~An excellent example is Georgia's 13th place in ETCC 2011 such that it have not played any teams better according to the official ranking.} such as when a Hungarian commentator speaks about '\emph{the curse of the Swiss system}'.\footnote{~See at \url{http://sakkblog.postr.hu/sokan-palyaznak-dobogos-helyezesre-izgalmas-utolso-fordulo-dont}.}
It is not necessarily the mistake of Swiss system rather a failure of the official ranking, which can be improved significantly by accounting for the strength of opponents.

Nevertheless, the choice between simplicity and more plausible rankings is not a modelling issue. An alternative may be to use these methods only for tie-breaking purposes.  

There are some obvious areas of future research.
In the analysis several complications observed have been neglected like matches played with black or white (an unavoidable problem in individual tournaments) or different number of matches due to byes or unplayed games. The choice of parameter $\varepsilon$ also requires further investigation.
Our findings can be strengthened or falsified by the examination of other competitions and simulations of Swiss system tournaments.

Finally, two possible uses of the proposed ranking method are worth to mention. First, it can be incorporated into the pairing algorithm, resulting in a more balanced schedules. Second, extensive analysis of the stability of a ranking between subsequent rounds may contribute to the choice of the number of rounds, which can be made endogenous as a function of the number of participants and other restrictions.

\section*{Acknowledgements}
\addcontentsline{toc}{section}{Acknowledgements}
\noindent
We are grateful to two anonymous referees for their valuable comments and suggestions. \\
The research was supported by OTKA grant K 111797 and by the MTA Premium Post Doctorate Research Program. \\
This research was partially supported by Pallas Athene Domus Scientiae Foundation. The views expressed are those of the author's and do not necessarily reflect the official opinion of Pallas Athene Domus Scientiae Foundation.


\end{document}